\documentclass[11pt]{article}
\usepackage{mathrsfs}
\usepackage{amssymb}
\usepackage{amsfonts}
\usepackage{color}
\usepackage{amsthm}
\usepackage{amsmath}
\usepackage{ascmac}
\usepackage[top=35truemm, bottom=30truemm, left=30truemm, right=30truemm]{geometry}
\usepackage{graphicx}
\usepackage{fancybox}
\newtheorem{dfn}{Definition}
\newtheorem{thm}{Theorem}[section]
\newtheorem{lem}[thm]{Lemma}
\newtheorem{cor}[thm]{Corollary}
\newtheorem{prop}[thm]{Proposition}

\title{Asymptotic property of current for a conduction model of \\ Fermi particles on finite lattice}
\author{Yamaga Kazuki}
\date{}

\begin{document}
\maketitle

\begin{abstract}
In this paper, we introduce a conduction model of Fermi particles on a finite sample, and investigate the asymptotic behavior of stationary current for large sample size. In our model a sample is described by a one-dimensional finite lattice on which Fermi particles injected at both ends move under various potentials and noise from the environment. We obtain a simple current formula. The formula has broad applicability and is used to study various potentials. When the noise is absent, it provides the asymptotic behavior of the current in terms of a transfer matrix. In particular, for dynamically defined potential cases, a relation between exponential decay of the current and the Lyapunov exponent of a relevant transfer matrix is obtained. For example, it is shown that the current decays exponentially for the Anderson model. On the other hand, when the noise exists but the potential does not, an explicit form of the current is obtained, which scales as $1/N$ for large sample size $N$. Moreover, we provide an extension to higher dimensional systems. For a three-dimensional case, it is shown that the current increases in proportion to cross section and decreases in inverse proportion to the length of the sample.
\end{abstract}

\section{Introduction}
A unified theory for nonequilibrium systems is still lacking, while statistical mechanics for equilibrium systems well-connects the microscopic and macroscopic world. 
This occurs mainly owing to the existence of various states in nonequilibrium systems. Therefore, it is important to consider a specific, physically interesting subclass of nonequilibrium states. Nonequilibrium stationary state induced by multiple thermal and particle reservoirs should be an important class, which has been studied for a long time \cite{spohn1978irreversible,ruelle2000natural,jakvsic2002mathematical}.
For example, \cite{aschbacher2007transport} considers electric conduction in mesoscopic systems as a problem of nonequilibrium stationary states of many body Fermi particle systems and derives the Landauer formula. In \cite{bruneau2016conductance}, the problem of how the structure of a sample between reservoirs determines the property of current is studied, and the equivalence of a ballistic transport and the existence of absolutely continuous spectrum is confirmed. Thus, if the absolutely continuous spectrum is empty, current goes to $0$ in the limit taking the sample size infinite.  There are many physically important models that do not have absolutely continuous spectrum such as the Anderson model and the Fibonacci Hamiltonian, which is considered as the one-dimensional model of a quasi-crystal. While the result is important, because the real sample size is finite, it is interesting to investigate the scaling of convergence. The scaling behavior depends on the sample structure. This problem has not been solved yet by the authors of \cite{bruneau2016conductance,bruneau2016absolutely}. In their model, a sample is connected to infinitely extended reservoirs at its ends; thus mathematical tools such as operator algebra and scattering theory are used. 

This study clarifies the problem of 'the scaling of the current' by introducing a simple finite dimensional conduction model. We focus on a sample described by a finite lattice on which many-body noninteracting Fermi particles are moving under various potentials and certain noise called dephasing noise from the environment. The exchange of particles between a sample and reservoirs is performed at the ends of the sample. See figure 1. This effect is described by a Lindblad-type generator. Because our model does not have an infinite part, the entire analysis is performed within linear algebra. The same model  has already been studied in \cite{prosen2008third,vznidarivc2010exact}. The difference is that we solve the time evolution using the approach of \cite{davies1977}. The following simple current formula is obtained,
\[ \mathcal{J}_\beta(N)=4(\alpha_{in}^l\alpha^r_{out}-\alpha^l_{out}\alpha_{in}^r)\int^\infty_0\langle e_1,T_s(p_N)e_1\rangle ds \]
where $(\alpha_{in}^l\alpha^r_{out}-\alpha^l_{out}\alpha_{in}^r)$ is a term determined by the strength of interaction at the both ends, and the integral is related with a two-point function which can be evaluated rigorously in various models. This formula can be applied to a wide class, which allows various types of potentials. Based on this formula, we consider how the scaling of the current is determined by potentials and noise. 
\begin{figure}[h]
    \centering
    \includegraphics[width=11cm]{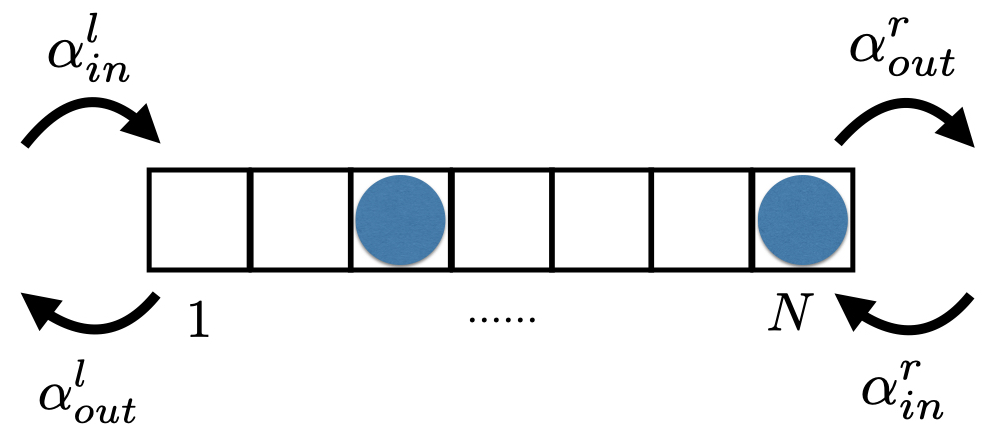}
    \caption{conduction model of Fermi particles on finite lattice}
    \label{fig:my_label}
\end{figure}

This paper is organized as follows. In section 2, we introduce a conduction model of Fermi particles on a one-dimensional lattice. By solving the time evolution of the two-point function, we show that it converges to a constant in the long time limit. In particular, as the current is described by a two-point function, we obtain a simple current formula described above. In section 3, we consider the asymptotic behavior of the current for large sample size, using the above formula. We first consider the noiseless case in subsection 3.1 and show that the current formula can be evaluated in terms of transfer matrix. This result shows that both our model and the model in \cite{bruneau2016conductance} give the same prediction for the asymptotic property of the current. In addition, in case of dynamically defined potential, such as the Anderson model and the Fibonacci Hamiltonian, the scaling of the asymptotic behavior is shown to be related with the Lyapunov exponent. In subsection 3.2, we introduce the noise called dephasing noise. We obtain an explicit form of the current if the potential is absent. The current decays scales as  $1/N$. This result coincides with that of \cite{vznidarivc2010exact}, which takes a different approach from ours. The scaling of the current for general potentials is not obtained yet. But, it is shown that for strong noise the main term of current decays as $1/N$, and the current may increase by adding strong noise to random systems. Section 4 is devoted to the generalization to higher dimensional systems.  The same formula as the one-dimensional system is obtained. If the noise exists and potential is absent, it is shown that the current increases in proportion to cross section and decreases in inverse proportion to the length of the sample. In the last section, we provide conclusions and discuss the related studies.

\section{Conduction model of Fermi particles on a one-dimensional finite lattice}
In this section, we introduce a conduction model of noninteracting Fermi particles on a one-dimensional finite lattice. First we consider the dynamics of a two-point function and its long time limit (2.1). Then in 2.2, we focus on current and obtain a simple current formula (Theorem 2.2).
\subsection{dynamics}
Let us consider a many body system of Fermi particles moving under various potentials and noise on a one-dimensional finite lattice $[1,N]\cap\mathbb{Z}$ ($N\in\mathbb{N}$). The one-particle Hilbert space is $\mathbb{C}^N$. Denote its standard basis by $\{e_n\}_{n=1}^N$. The many body system is described by the creation and annihilation operators $a^*(f),a(f)$, where $f\in\mathbb{C}^N$. Let us write $a^\#_n$ for $a^\#(e_n)$ as usual ($a^\#$ is $a^*$ or $a$). These operators satisfy the following canonical anti-commutation relations:
\[ \{a^*_i,a_j\}=\delta_{ij}I,\hspace{10pt} \{a_i,a_j \}=0, \]
where $\{A,B\}=AB+BA$, $\delta_{ij}$ is the Kronecker delta and $I$ is the identity operator on $\mathbb{C}^N$. In the sequel, we write shortly $c\times I$ as $c$ ($c\in\mathbb{C}$). Suppose that the total Hamiltonian is 
\[ H=\sum_{n=1}^{N-1}\left[-(a_n^*a_{n+1}+a_{n+1}^*a_n)+v(n)a^*_na_n\right], \]
where $v(\cdot)$ is a real-valued function called potential. Since we will consider the limit $N\to\infty$, potential $v$ is given as a bounded function on $\mathbb{N}$. Let $\mathcal{A}$ be the algebra generated by the creation and annihilation operators, and $\theta\colon\mathcal{A}\to\mathcal{A}$ be a *-automorphism determined by $\theta(a_n^\#)=-a_n^\#$.

For real numbers $\alpha_{in}^l,\alpha_{out}^l,\alpha_{in}^r,\alpha_{out}^r,\beta$ greater than or equal to $0$ (at least one of $\alpha_{in}^l,\alpha_{out}^l,\alpha_{in}^r,\alpha_{out}^r$ is not $0$), define a linear map $L\colon\mathcal{A}\to\mathcal{A}$ as
\begin{eqnarray*}
 L(A)&=&i[H,A]+\\
&&\alpha_{in}^l(2a_1\theta(A)a_1^*-\{a_1a_1^*,A\})+\alpha_{out}^l(2a_1^*\theta(A)a_1-\{a_1^*a_1,A\}) \\
&&+\alpha_{in}^r(2a_N\theta(A)a_N^*-\{a_Na_N^*,A\})+\alpha_{out}^r(2a_N^*\theta(A)a_N-\{a_N^*a_N,A\}) \\
&&+\beta\sum_{n=1}^N\left(a^*_na_n Aa^*_na_n-\frac{1}{2}\{a^*_na_n,A\}\right).
\end{eqnarray*}
Here, $[H,A]=HA-AH$.

It is obvious from the form of each term that $L$ generates a Quantum Dynamical Semigroup $\{e^{tL}\}_{t\ge0}$ on $\mathcal{A}$ \cite{lindblad1976generators}. That is, $e^{tL}$ is a CP (completely positive) map preserving identity (state transformation) for every $t\ge0$. The physical meaning of each term is as follows:

\underline{$i[H,A]$}\\
This term represents the Hamiltonian dynamics of many particles moving independently by a one-particle Hamiltonian
\[ (h\psi)(n)=-\psi(n+1)-\psi(n-1)+v(n)\psi(n) \]
($\psi(0)=\psi(N+1)=0$). It operates as
\[ i[H,a^*(f)a(g)]=ia^*(hf)a(g)-ia^*(f)a(hg). \]

\underline{The terms with coefficients $\alpha_{in}^l,\alpha_{out}^l,\alpha_{in}^r,\alpha_{out}^r$}\\
These terms represent the effects of adding a particle to site 1, removing from site 1, adding to site $N$ and removing from site $N$, respectively. Put $p_n=|e_n\rangle\langle e_n|\ (n=1,2,\cdots,N)$, 1-rank projections corresponding to the basis $\{e_n\}_{n=1}^N$, then these terms operate as
\[ 2a_1\theta(a(f)a^*(g))a_1^*-\{a_1a_1^*,a(f)a^*(g)\}=-a(p_1f)a^*(g)-a(f)a^*(p_1g) ,\]
\[ 2a_1^*\theta(a^*(f)a(g))a_1-\{a_1^*a_1,a^*(f)a(g)\}=-a^*(p_1f)a(g)-a^*(f)a(p_1g) \]
(replace $p_1$ by $p_N$ in the case of $N$). The dynamics generated by these terms and $i[H,A]$ is a special case of those in \cite{prosen2008third,davies1977}.

\underline{The term with coefficient $\beta$ (dephasing noise)}\\
This term represents noise from the environment called dephasing noise. Dephasing noise preserves the number of particles and destroys the coherence. Let us check this property. Denote
\[ D_n(A)=2a^*_na_n Aa^*_na_n-\{a^*_na_n,A\} \]
($n=1,2,\cdots,N$), then since $D_n$ commutes with each other, we have
\[ e^{t\sum_{n=1}^ND_n}=\prod_{n=1}^Ne^{tD_n}. \]
Easy calculation shows that
\[ D_n(a^*_ia_j)=\begin{cases}
0 & (i=j=n,\ i,j\neq n) \\
-a^*_ia_j & (\mathrm{otherwise}),
\end{cases} \]
\[ e^{tD_n}(a^*_ia_j)=\begin{cases}
a^*_ia_j & (i=j=n,\ i,j\neq n) \\
e^{-t}a^*_ia_j & (\mathrm{otherwise}).
\end{cases} \]
Recall that for every state $\omega$ on $\mathcal{A}$, its two-point function is described by a positive operator on $\mathbb{C}^N$: there is an operator $R\colon\mathbb{C}^N\to\mathbb{C}^N$ such that $0\le R\le I$ and
\[ \omega(a^*_ia_j)=\langle e_j,Re_i\rangle. \]
In the one-particle system, define a linear map $d_n\colon M_N(\mathbb{C})\to M_N(\mathbb{C})$ as
\[ d_n(a)=2p_nap_n-\{p_n,a\},\ a\in M_N(\mathbb{C})\]
($M_N(\mathbb{C})$ is the set of $N\times N$ complex matrices), then
\[ e^{td_n}(a)=a+(1-e^{-t})d_n(a). \]
If $i=j=n$ or $i,j\neq n$,
\[ \langle e_j,e^{td_n}(R)e_i\rangle=\langle e_j,Re_i\rangle=\omega(e^{tD_n}(a^*_ia_j)), \]
and if one of $i,j$ is $n$,
\[ \langle e_j,e^{td_n}(R)e_i\rangle=\langle e_j,Re_i\rangle+(1-e^{-t})\langle e_j,(-R)e_i\rangle=\omega(e^{tD_n}(a^*_ia_j)). \]
Therefore, the dynamics of the two-point function is described by $\displaystyle\prod^N_{n=1}e^{td_n}(R)$:
\[ \omega\left(e^{t\sum_{n=1}^ND_n}(a^*_ia_j)\right)=\omega\left(\prod_{n=1}^Ne^{tD_n}(a^*_ia_j)\right)=\left\langle e_j,\prod_{n=1}^Ne^{td_n}(R)e_i\right\rangle=\left\langle e_j,e^{t\sum_{n=1}^Nd_n}(R)e_i\right\rangle. \] 
Set $d(a)=\displaystyle\frac{1}{2}\sum_{n=1}^Nd_n(a)=\sum_{n=1}p_nap_n-a$, then $d^2=-d$ and
\[ e^{td}(a)=a+(1-e^{-t})d(a)=e^{-t}a+(1-e^{-t})\sum_{n=1}^Np_nap_n .\]
The pure state $|\psi\rangle\langle\psi|$ is transformed to
\[ e^{-t}|\psi\rangle\langle\psi|+(1-e^{-t})\sum_{n=1}^N|\psi(n)|^2p_n .\]
Thus, this dynamics destroys the coherence and transform a state to a convex combination of localized states $p_n$.

$L$ consists of the above three types of terms. Stationary current induced by the dynamics $e^{tL}$ is the main topic in this paper. From the above discussions it turns out that the dynamics of the two-point function is described by that of the one-particle system. Suppose that the two-point function of the state $\omega\circ e^{tL}$ is expressed as
\[ \omega\circ e^{tL}(a^*_ia_j)=\langle e_j,R(t)e_i\rangle, \]
then by calculating 
\[ \frac{d}{dt}\omega\circ e^{tL}(a^*_ia_j)=\omega\circ e^{tL}(L(a^*_ia_j)), \]
we obtain the following differential equation for $R(t)$: 
\begin{eqnarray*}
\frac{d}{dt}R(t)&=&-i[h,R(t)]-\{(\alpha_{in}^l+\alpha_{out}^l)p_1+(\alpha_{in}^r+\alpha_{out}^r)p_N,R(t)\}\\
&&+\beta\left(\sum_{n=1}^Np_nR(t)p_n-R(t)\right)+2\alpha_{in}^l p_1+2\alpha_{in}^r p_N  ,
\end{eqnarray*}
\[ R(0)=R. \]
It is easy to check that
\[ R(t)=T_t(R)+\int^t_0T_s(2\alpha_{in}^l p_1+2\alpha_{in}^r p_N)ds \]
is a solution of this equation, where $T_t$ is an operator semigroup on $M_N(\mathbb{C})$ generated by 
\[ l\colon a\mapsto-i[h,a]-\{(\alpha_{in}^l+\alpha_{out}^l)p_1+(\alpha_{in}^r+\alpha_{out}^r)p_N,a\}+\beta\left(\sum_{n=1}^Np_nap_n-a\right). \]
$T_t=e^{tl}$ is a CP map which does not preserve identity.

Let us consider the long time limit $t\to\infty$. In the case where $\beta=0$,
\[ T_t(a)=e^{-ith_D}ae^{ith_D^*} \]
for
\[ h_D=h-i(\alpha_{in}^l+\alpha_{out}^l)p_1-i(\alpha_{in}^r+\alpha_{out}^r)p_N. \]
Since the imaginary part of every eigenvalue of $h_D$ is less than $0$, $\displaystyle\lim_{t\to\infty}e^{-ith_D}=0$. Thus, we get
\[ \lim_{t\to\infty}R(t)=\int^\infty_0T_s(2\alpha_{in}^l p_1+2\alpha_{in}^r p_N)ds\equiv R_\infty .\]
The integral of the right hand side converges, because
\begin{eqnarray*}
0\le\int^{t_1}_{t_2}T_s(2\alpha_{in}^l p_1+2\alpha_{in}^r p_N)ds&\le&\int^{t_1}_{t_2}T_s(2(\alpha_{in}^l+\alpha_{out}^l)p_1+2(\alpha_{in}^r+\alpha_{out}^r)p_N)ds \\
&=&\left[T_s(I)\right]^{t_1}_{t_2}\to0\ (t_1,t_2\to\infty). 
\end{eqnarray*}

Note that $R_\infty$ does not depend on $R$. This means that whatever the initial state is, the two-point function converges to the same value $\langle e_j,R_\infty e_i\rangle$. Moreover, it can be shown that every state converges to the quasi free state determined by this two-point function \cite{davies1977}.

In the case where $\beta>0$, we have the same result for the two-point function.
\begin{thm}
$\displaystyle\lim_{t\to\infty}T_t=\lim_{t\to\infty}e^{tl}=0$.
\end{thm}
\begin{proof}
Recall that $M_N(\mathbb{C})$ is a Hilbert space for the Hilbert-Schmidt inner product, $\langle a,b\rangle_{HS}=\mathrm{Tr}a^*b$. Let us decompose the generator of $l\colon M_N(\mathbb{C})\to M_N(\mathbb{C})$ as $l=-iX-Y-\beta Z$ for $X,Y,Z\colon M_N(\mathbb{C})\to M_N(\mathbb{C})$ defined as
\[ Xa=[h,a] \]
\[ Ya=\{(\alpha_{in}^l+\alpha_{out}^l)p_1+(\alpha_{in}^r+\alpha_{out}^r)p_N,a\} \]
\[ Za=a-\sum_{n=1}^Np_nap_n .\]
$X,Y,Z$ are self-adjoint and especially $Y,Z$ are positive. Let us check that $Z$ is positive: 
\[ \langle a,Za\rangle_{HS}=\mathrm{Tr}\left(a^*a-a^*\sum_{n=1}^Np_nap_n\right)=\sum_{n=1}^N\mathrm{Tr}p_na^*(I-p_n)ap_n\ge0,\ ^\forall a\in M_N(\mathbb{C}). \]
Let $x\in\mathbb{C}$ be an eigenvalue of $l$ and $a\in M_N(\mathbb{C})$ be a corresponding unit eigenvector, that is, $x,a$ satisfy
\[ la=xa ,\hspace{10pt}\langle a,a\rangle_{HS}=1. \]
By $l=-iX-Y-\beta Z$, $\mathrm{Re}x$, the real part of $x$, satisfies that
\[ \mathrm{Re}x=-\langle a,(Y+\beta Z)a\rangle\le0 .\]
If $\mathrm{Re}x=0$, we have
\begin{equation}
\langle a,Ya\rangle=0 ,
\end{equation}
\begin{equation}
\langle a,Za\rangle=0, 
\end{equation} 
since $Y$ and $Z$ are positive. By equation(2),
\[ \mathrm{Tr}p_na^*(I-p_n)ap_n=0,\ n=1,2,\cdots,N ,\]
\[ \to (I-p_n)ap_n=0 .\]
Thus, $a$ is diagonalized for the basis $\{e_n\}_{n=1}^N$ (we write its entry as $a_{ij}$). Assume that $\alpha_{in}^l+\alpha_{out}^l>0$ (otherwise $\alpha_{in}^r+\alpha_{out}^r>0$ must hold and repeat the following processes from $N$ instead of $1$), then by equation(1)
\[ a_{11}=\mathrm{Tr}p_1ap_1=0 .\]
Since $a$ is diagonalized,
\[ 0=xa_{12}=(la)_{12}=ia_{11}+ia_{13}-ia_{22}-(\alpha_{in}^l+\alpha_{out}^l+\beta)a_{12}=-ia_{22}, \]
\[ 0=xa_{23}=(la)_{23}=ia_{22}+ia_{24}-ia_{33}-ia_{13}-(\alpha_{in}^l+\alpha_{out}^l)a_{23}=-ia_{33}. \]
Repeat these processes until $0=(la)_{N-1\ N}$, then we finally get $a_{11}=a_{22}=\cdots=a_{NN}=0$. This implies that $a=0$. However $a=0$ contradicts to the assumption that $\langle a,a\rangle_{HS}=1$. Thus, $\mathrm{Re}x<0$ must hold for every eigenvalue of $l$ and 
\[ \lim_{t\to\infty}e^{tl}=0. \]
\end{proof}
By this theorem, in the case where $\beta>0$ we also have
\[ \lim_{t\to\infty}R(t)=\int^\infty_0T_s(2\alpha_{in}^l p_1+2\alpha_{in}^r p_N)ds\equiv R_\infty . \]

\subsection{current formula}
In this subsection we will focus on current. Since current is expressed by two-point function, it converges to a constant in the limit $t\to\infty$. We will consider how the sign of the current is determined by the relation of constants $\alpha_{in}^l,\alpha_{out}^l,\alpha_{in}^r,\alpha_{out}^r$. The current is shown to be expressed by a simple formula (Theorem 2.2). 

At first, recall that the observable of current from site $n$ to $n+1$ is 
\[ j_n=-i(a_n^*a_{n+1}-a^*_{n+1}a_n). \]
As shown in the previous subsection, for any state $\omega$ the limit $\displaystyle\lim_{t\to\infty}\omega\circ e^{tL}(j_n)$ exists and is independent of $\omega$. In fact it does not depend on $n$. Let us check it. By the definition of generator, for any $\epsilon>0$ there is $h>0$ such that
\[ \left\|\frac{e^{hL}(a_n^*a_n)-a^*_na_n}{h}-L(a^*_na_n)\right\|<\epsilon. \]
Thus we have
\[ |\omega\circ e^{tL}(L(a^*_na_n))|<\epsilon+\left|\omega\circ e^{tL}\left(\frac{e^{hL}(a_n^*a_n)-a^*_na_n}{h}\right)\right|,\ ^\forall t\ge0 \]
and $\displaystyle\limsup_{t\to\infty}|\omega\circ e^{tL}(L(a^*_na_n))|\le\epsilon$. Since $\epsilon$ is arbitrary, $\displaystyle\lim_{t\to\infty}\omega\circ e^{tL}(L(a^*_na_n))=0$. This equation and 
\[ L(a^*_na_n)=-i(a_{n-1}^*a_n-a^*_na_{n-1})+i(a^*_na_{n+1}-a^*_{n+1}a_n)=j_{n-1}-j_n ,\ n=2,\cdots,N-1\]
show that the limit of the current does not depend on $n$. We denote the limit of current by $\mathcal{J}_\beta(N)=\displaystyle\lim_{t\to\infty}\omega\circ e^{tL}(j_1)$ (it depends on the sample size $N$). Then it is expressed as
\begin{eqnarray*}
\mathcal{J}_\beta(N)&=&-i\langle e_2,R_\infty e_1\rangle +i\langle e_1,R_\infty e_2\rangle  \\
&=&2\mathrm{Im}\langle e_2,R_\infty e_1\rangle  .
\end{eqnarray*}

$\mathcal{J}_\beta(N)$ has the following simple expression. This is one of our main results in this paper. 
\begin{thm}
\[ \mathcal{J}_\beta(N)=4(\alpha_{in}^l\alpha_{out}^r-\alpha_{out}^l\alpha_{in}^r)\int^\infty_0\langle e_1,T_t(p_N)e_1\rangle dt=-4(\alpha_{in}^l\alpha_{out}^r-\alpha_{out}^l\alpha_{in}^r)\langle e_1,l^{-1}(p_N)e_1\rangle. \]
\end{thm}
\begin{proof}
By $2\mathrm{Im}|e_1\rangle\langle e_2|=-i[h,p_1]$ and the definition of $R_\infty$,
\begin{eqnarray*}
\mathcal{J}_\beta(N)&=&-\int^\infty_0\mathrm{Tr}i[h,p_1]e^{tl}(2\alpha_{in}^l p_1+2\alpha_{in}^r p_N)dt \\
&=&-\int^\infty_0\mathrm{Tr}p_1le^{tl}(2\alpha_{in}^l p_1+2\alpha_{in}^r p_N)dt-2(\alpha_{in}^l+\alpha_{out}^l)\int^\infty_0\mathrm{Tr}p_1e^{tl}(2\alpha_{in}^l p_1+2\alpha_{in}^r p_N)dt \\
&=&2\alpha_{in}^l\mathrm{Tr}p_1-4\alpha_{in}^l(\alpha_{in}^l+\alpha_{out}^l)\int^\infty_0\mathrm{Tr}p_1e^{tl}(p_1)dt-4\alpha_{in}^r(\alpha_{in}^l+\alpha_{out}^l)\int^\infty_0\mathrm{Tr}p_1e^{tl}(p_N)dt .
\end{eqnarray*}
By the equation
\[ I=-\left[e^{tl}(I)\right]^\infty_0=\int^\infty_0e^{tl}(2(\alpha_{in}^l+\alpha_{out}^l)p_1+2(\alpha_{in}^r+\alpha_{out}^r)p_N)dt, \]
we have
\[ \alpha_{in}^l\mathrm{Tr}p_1-2\alpha_{in}^l(\alpha_{in}^l+\alpha_{out}^l)\int^\infty_0\mathrm{Tr}p_1e^{tl}(p_1)dt=2\alpha_{in}^l(\alpha_{in}^r+\alpha_{out}^r)\int^\infty_0\mathrm{Tr}p_1e^{tl}(p_N)dt .\]
Combining these equations, we get
\[ \mathcal{J}_\beta(N)=4(\alpha_{in}^l\alpha_{out}^r-\alpha_{out}^l\alpha_{in}^r)\int^\infty_0\mathrm{Tr}p_1e^{tl}(p_N)dt .\]
In order to obtain the latter equation of the theorem, we use the well-known formula for operator semigroups: for any $\epsilon>0$
\[ \int^\infty_0e^{-\epsilon t}e^{tl}dt=(\epsilon-l)^{-1} \]
holds. As discussed before, the real part of every eigenvalue of $l$ is less than $0$. This implies that $\mathrm{ker}l=\{0\}$ and $l$ is invertible. Thus, we get
\begin{eqnarray*}
\int^\infty_0\langle e_1,T_t(p_N)e_1\rangle dt&=&\lim_{\epsilon\downarrow0}\int^\infty_0e^{-\epsilon t}\langle e_1,T_t(p_N)e_1\rangle dt \\
&=&\lim_{\epsilon\downarrow0}\langle e_1,(\epsilon-l)^{-1}(p_N)e_1\rangle \\
&=&-\langle e_1,l^{-1}(p_N)e_1\rangle .
\end{eqnarray*}
\end{proof}
Since $\int^\infty_0\langle e_1,T_t(p_N)e_1\rangle dt>0$, the sign of $\mathcal{J}_\beta(N)$ is completely determined by the coefficient $\alpha_{in}^l\alpha_{out}^r-\alpha_{out}^l\alpha_{in}^r$. Let us check that $\int^\infty_0\langle e_1,T_t(p_N)e_1\rangle dt>0$.
\begin{eqnarray*}
\langle e_1,T_t(p_N)e_1\rangle&=&\sum_{n=0}^\infty\left\langle e_1,\frac{(tl)^n}{n!}(p_N)e_1\right\rangle \\
&=&\left\langle e_1,\frac{(tl)^{2N}}{(2N)!}(p_N)e_1\right\rangle+\sum_{n=2N+1}^\infty\left\langle e_1,\frac{(tl)^n}{n!}(p_N)e_1\right\rangle \\
&=&\frac{t^{2N}}{(2N)!}{}_{2N}\mathrm{C}_N\langle e_1,h^Np_Nh^Ne_1\rangle+\sum_{n=2N+1}^\infty\left\langle e_1,\frac{(tl)^n}{n!}(p_N)e_1\right\rangle \\
&=&\frac{t^{2N}}{(N!)^2}+\sum_{n=2N+1}^\infty\left\langle e_1,\frac{(tl)^n}{n!}(p_N)e_1\right\rangle.
\end{eqnarray*}
Since 
\[ \left|\sum_{n=2N+1}^\infty\left\langle e_1,\frac{(tl)^n}{n!}(p_N)e_1\right\rangle\right|\le t^{2N+1}\sum_{n=2N+1}^\infty\left\langle e_1,\frac{l^n}{n!}(p_N)e_1\right\rangle \]
holds for $0\le t<1$, for sufficiently small $t>0$ we have
\[ \langle e_1,T_t(p_N)e_1\rangle>0. \]

\section{Asymptotic behavior of current}
In the previous section, we obtained a current formula applicable in general settings (Theorem 2.2). In this section, using this formula, we investigate how potentials and noise determine the asymptotic behavior of the current $\mathcal{J}_\beta(N)$ for large sample size $N$. Since we would like to consider the situation that the current $\mathcal{J}_\beta(N)$ is not $0$, let $\alpha_{in}^l+\alpha_{out}^l,\alpha_{in}^r+\alpha_{out}^r>0$. At first, we deal with the noiseless case ($\beta=0$). And next, the case where $\beta>0$, mainly $v=0$, is considered.

\subsection{$\beta=0$ : noiseless case}
In this subsection, we first prove the following proposition, which is applicable to arbitrary potentials.
\begin{prop}
\[ -\langle e_1,l^{-1}(p_N)e_1\rangle
=\frac{1}{2\pi}\int_\mathbb{R}\left|\langle e_1,(h_D-E)^{-1}e_N\rangle\right|^2dE .\]
\end{prop}
Using this formula, we relate the current $\mathcal{J}_\beta(N)$ to transfer matrix. In addition, in case of dynamically defined potential, such as the Anderson model, the scaling of the asymptotic behavior is shown to be related with the Lyapunov exponent. 

Recall that in noiseless case
\[ la=-i(h_Da-ah^*_D). \]
As mentioned before, the imaginary part of every eigenvalue of $h_D$ is less than 0, and $h_D$ is invertible. Let us prepare a lemma.
\begin{lem}
For $V\in M_N(\mathbb{C})$, define a linear map $g_V\colon M_N(\mathbb{C})\to M_N(\mathbb{C})$ as
\[ g_V(x)=Vx-xV^*,\ x\in M_N .\]
Suppose that $V$ is invertible and the imaginary part of every eigenvalue is less than 0. Then, $g_V$ is invertible and 
\[ g_V^{-1}(x)=\frac{i}{2\pi}\int_\mathbb{R}(E-V)^{-1}x(E-V^*)^{-1}dE .\]
\end{lem}
\begin{proof}
Since the integrand (operator) in the right hand side is continuous for $E$ and 
\[ \|(V-E)^{-1}\|\le \frac{1}{|E|-\|V\|} \]
for $E$ with large absolute value, the integral converges and defines a linear map on $M_N(\mathbb{C})$. Let us denote it by $h(x)$. Since
\[ g_V(x)=(V-E)x-x(V^*-E) \]
for $E\in\mathbb{R}$, we have
\[ h\circ g_V(x)=-\frac{i}{2\pi}\int_\mathbb{R}x(E-V^*)^{-1}dE+\frac{i}{2\pi}\int_\mathbb{R}(E-V)^{-1}xdE . \]
Let us consider the entry of the matrix
\[ \int_\mathbb{R}(E-V)^{-1}dE .\]

In general, $(i,j)$-entry of the inverse matrix of an $N\times N$ matrix $A=(a_{ij})_{i,j=1}^N$ is expressed as
\[ \mathrm{det}(A)^{-1}(-1)^{i+j}\mathrm{det}(A_{ji}). \]
$A_{ij}$ is an $(N-1)\times (N-1)$ matrix called factor matrix, which is made by removing the i-th row and the j-th column from $A$. $\mathrm{det}(E-V)$ is a polynomial that has degree of $N$ and the coefficient of $E^N$ is 1. Let us write $\mathrm{det}(E-V)=E^N+a_N(E)$. Set $A=E-V$, then $\mathrm{det}(A_{ii})$ is a polynomial with degree of $N-1$ and the coefficient of $E^{N-1}$ is 1. Let us write $\mathrm{det}(A_{ii})=E^{N-1}+b_N^i(E)$. If $i\neq j$, then $\mathrm{det}(A_{ij})$ is a polynomial with degree of $N-2$ and denoted by $c_N^{ij}(E)$. Define $\mathbb{C}_+=\{z\in\mathbb{C}\mid\mathrm{Im}z\ge0\}$. Since $\mathrm{det}(E-V)$ has no zeros in $\mathbb{C}_+$, $(E-V)^{-1}_{ij}$ is regular in a region containing $\mathbb{C}_+$ (note that $(E-V)^{-1}_{ij}$ is the $(i,j)$-entry of $(E-V)^{-1}$, not factor matrix). For $R>0$, define a cycle $\Gamma_R$ as $\{z\in\mathbb{C}\mid\mathrm{Im}z=0, \mathrm{Re}z\in[-R,R]\}\cup\{Re^{i\theta}\mid \theta\in[
0,\pi]\}$, then
\[ \oint_{\Gamma_R}(E-V)^{-1}_{ij}dE=0 .\]
Set $C_R=\{Re^{i\theta}\mid \theta\in[0,\pi]\}$. \\
(i)$i=j$\\
\begin{eqnarray*}
\int_{C_R}(z-V)^{-1}_{ii}dz&=&\int_{C_R}\frac{z^{N-1}+b^i_N(z)}{z^N+a_N(z)}dz \\
&=&\int_0^{\pi}\frac{R^{N-1}e^{i(N-1)\theta}+b_N^i(Re^{-i\theta})}{R^Ne^{iN\theta}+a_N(Re^{i\theta})}iRe^{i\theta}d\theta \\
&=&i\int_0^{\pi}\frac{R^N+e^{-i(N-1)\theta}b_N^i(Re^{-i\theta})}{R^N+e^{-iN\theta}a_N(Re^{i\theta})}d\theta.
\end{eqnarray*}
This converges to $i\pi$ as $R\to\infty$. By
\[ \int^R_{-R}(E-V)^{-1}_{ii}dE+\int_{C_R}(z-V)^{-1}_{ii}dz=\oint_{\Gamma_R}(E-V)^{-1}_{ij}dE=0, \]
we have
\[ \int_\mathbb{R}(E-V)^{-1}_{ii}dE=\lim_{R\to\infty}\int^R_{-R}(E-V)^{-1}_{ii}dE=-i\pi .\]
(ii)$i\neq j$\\
\begin{eqnarray*}
\int_{C_R}(z-V)^{-1}_{ij}dz&=&\int_{C_R}\frac{c_N^{ij}(z)}{z^N+a_N(z)}dz \\
&=&i\int_0^{\pi}\frac{Re^{-i(N-1)\theta}c_N^{ij}(Re^{i\theta})}{R^N+e^{-iN\theta}a_N(Re^{i\theta})}d\theta.
\end{eqnarray*}
This converges to 0 as $R\to\infty$. Thus,
\[ \int_\mathbb{R}(E-V)^{-1}_{ij}dE=0. \]
Summarizing the above calculations, we get
\[ \int_\mathbb{R}(E-V)^{-1}dE=-i\pi I \]
and
\[ h\circ g_V(x)=-\frac{i}{2\pi}i\pi x+\frac{i}{2\pi}(-i\pi)x=x. \]
This implies that $g_V$ is an injection. Since the space that $g_V$ operate is finite dimensional, $g_V$ is also surjective. Therefore, $g_V$ is invertible and
\[ g_V^{-1}(x)=h(x)=\frac{i}{2\pi}\int_\mathbb{R}(E-V)^{-1}x(E-V^*)^{-1}dE .\]
\end{proof}

Applying this lemma for $V=h_D$, then we obtain Proposition 3.1:
\begin{eqnarray*}
-\langle e_1,l^{-1}(p_N)e_1\rangle 
&=&\frac{1}{2\pi}\int_\mathbb{R}\langle e_1,(h_D-E)^{-1}p_N(h_D^*-E)^{-1}e_1\rangle dE\\
&=&\frac{1}{2\pi}\int_\mathbb{R}\left|\langle e_1,(h_D-E)^{-1}e_N\rangle\right|^2dE .
\end{eqnarray*}
By this equation, in order to know the asymptotic behavior of the current $\mathcal{J}_\beta(N)$, we have to investigate that of $\left|\langle e_1,(h_D-E)^{-1}e_N\rangle\right|^2$. As we will see in the following, $\left|\langle e_1,(h_D-E)^{-1}e_N\rangle\right|$ is related to transfer matrix.

Let us recall transfer matrix. Although we are considering a system on finite lattice $[1,N]\cap\mathbb{N}$, potential is given as a function $v\colon\mathbb{N}\to\mathbb{R}$ in order to take limit $N\to\infty$. For $E\in\mathbb{C}$, if $\psi\in\mathbb{C}^N$ satisfies
\[ h\psi=E\psi, \]
then the relation
\[ \left(
\begin{array}{c}
\psi(n+1) \\
\psi(n)
\end{array}\right) = 
\left(
\begin{array}{cc}
v(n)-E & -1 \\
1 & 0 
\end{array}\right)
\left(
\begin{array}{c}
\psi(n) \\
\psi(n-1)
\end{array}\right),\ n=1,\cdots,N \]
holds (here, $\psi(0)=\psi(N+1)=0$). A $2\times 2$ matrix
\[ T_N(E)\equiv\left(
\begin{array}{cc}
v(N)-E & -1 \\
1 & 0 
\end{array}\right)\cdots
\left(
\begin{array}{cc}
v(1)-E & -1 \\
1 & 0 
\end{array}\right) \]
is called a transfer matrix. It is in $SL(2,\mathbb{C})$ and thus $\|T_N(E)\|\ge1$.

For $E\in\mathbb{R}$, define
\[ g_{ij}(E)=\langle e_i,(h_D-E)^{-1}e_j\rangle .\]
These values are related to transfer matrix as follows.

\begin{lem}
\[ \tilde{T}_N(E)\left(
\begin{array}{cc}
g_{11}(E) & g_{1N}(E)  \\
1 & 0
\end{array}\right)=\left(
\begin{array}{cc}
0 & 1 \\
g_{N1}(E) & g_{NN}(E)
\end{array}\right) ,\]
where $\tilde{T}_N(E)$ is a transfer matrix corresponding to a complex-valued potential $\tilde{v}$ defined as $\tilde{v}(1)=v(1)-i(\alpha_{in}^l+\alpha_{out}^l)$, $\tilde{v}(N)=v(N)-i(\alpha_{in}^r+\alpha_{out}^r)$ and $\tilde{v}(n)=v(n)$ for $n=2,\cdots,N-1$.
\end{lem}

We do not give the proof here, since it is in \cite{bruneau2013landauer} (Lemma 2.2). By this lemma, we can evaluate $\left|\langle e_1,(h_D-E)^{-1}e_N\rangle\right|$ using transfer matrix.

\begin{lem}
There is a constant $M>0$ independent of $E\in\mathbb{R},\ N\in\mathbb{N}$ such that
\[ |g_{ij}(E)|\le M \]
($i,j=1,N$). There is a constant $K$ such that
\[ \frac{1}{\|\tilde{T}_N(E)\|}\le\left|\langle e_1,(h_D-E)^{-1}e_N\rangle\right|\le\frac{K}{\|\tilde{T}_N(E)\|} .\]
\end{lem}
\begin{proof}
By resolvent formula,
\begin{eqnarray*}
 &&(\alpha_{in}^l+\alpha_{out}^l)|g_{11}(E)|^2+(\alpha_{in}^r+\alpha_{out}^r)|g_{N1}(E)|^2\\
 &=&\left\langle e_1,(h_D^*-E)^{-1}\{(\alpha_{in}^l+\alpha_{out}^l)p_1+(\alpha_{in}^r+\alpha_{out}^r)p_N\}(h_D-E)^{-1}e_1\right\rangle \\
&=&\frac{1}{2i}(g_{11}(E)-\overline{g_{11}(E)}) \\
&\le&|g_{11}(E)| .
\end{eqnarray*}
From this inequality, we have
\begin{eqnarray} 
|g_{11}(E)|^2-\frac{1}{\alpha_{in}^l+\alpha_{out}^l}|g_{11}(E)|\le 0,\\ (\alpha_{in}^r+\alpha_{out}^r)|g_{N1}(E)|^2\le-(\alpha_{in}^l+\alpha_{out}^l)|g_{11}(E)|^2+|g_{11}(E)| . 
\end{eqnarray}
By inequality(3),
\[ |g_{11}(E)|\le\frac{1}{\alpha_{in}^l+\alpha_{out}^l} \]
and by inequality(4),
\[ |g_{N1}(E)|^2\le\frac{1}{4(\alpha_{in}^l+\alpha_{out}^l)(\alpha_{in}^r+\alpha_{out}^r)} .\]
Similarly, we get
\[ |g_{1N}(E)|^2\le\frac{1}{4(\alpha_{in}^l+\alpha_{out}^l)(\alpha_{in}^r+\alpha_{out}^r)},\ |g_{NN}(E)|\le\frac{1}{\alpha_{in}^r+\alpha_{out}^r}. \]
The former inequality of the lemma is obtained.

Operating both hand sides of the equation of Lemma 3.3 to a vector $\left(\begin{array}{c}0 \\ 1\end{array}\right)$, we obtain
\[ 1\le\left\|\left(\begin{array}{c}1 \\ g_{NN}(E)\end{array}\right)\right\|\le\|\tilde{T}_N(E)\||g_{1N}(E)| .\]
Since $g_{1N}(E)$ is not $0$,
\[ \left(
\begin{array}{cc}
g_{11}(E) & g_{1N}(E) \\
1 & 0
\end{array}\right) \]
is invertible and by Lemma 3.3 we have
\begin{eqnarray*}
 g_{1N}(E)\tilde{T}_N(E)&=&g_{1N}(E)\left(
\begin{array}{cc}
0 & 1  \\
g_{N1}(E) & g_{NN}(E)
\end{array}\right)\left(
\begin{array}{cc}
g_{11}(E) & g_{1N}(E) \\
1 & 0
\end{array}\right)^{-1} \\
&=&\left(
\begin{array}{cc}
0 & 1  \\
g_{N1}(E) & g_{NN}(E)
\end{array}\right)\left(
\begin{array}{cc}
0 & -g_{1N}(E) \\
-1 & g_{11}(E)
\end{array}\right).
\end{eqnarray*}
Since all the entries of the right hand side are bounded, the norm is also bounded by an $E,N$-independent constant $K$: 
\[ \|\tilde{T}_N(E)\|\le\frac{K}{|g_{1N}(E)|}. \]
\end{proof}

Easy calculation shows that there are $E,N$-independent constants $a,b>0$ such that
\[ a\|T_N(E)\|\le\|\tilde{T}_N(E)\|\le b\|T_N(E)\| .\]
Therefore, the asymptotic behavior of the current is determined by that of 
\begin{equation}
\int_\mathbb{R}\frac{1}{\|T_N(E)\|^2}dE .
\end{equation}

Denote $C=2+\displaystyle\sup_{n\in\mathbb{N}}|v(n)|$, then the spectrum of $h$, $\sigma(h)$, is contained in the interval $[-C,C]$. Set $R>C+1$. The following facts show that the integral over large energy decays so rapidly that we do not have to care when considering the asymptotic behavior. This is used when we consider concrete models later.

\begin{thm}
\[ \liminf_{N\to\infty}\left(-\frac{1}{N}\log\int^\infty_R\frac{dE}{\|T_N(E)\|^2}\right)\ge2\log(R-C)>0 .\]
It is same for
\[\int^{-R}_{-\infty}\frac{dE}{\|T_N(E)\|^2} .\]
\end{thm}

By this theorem, we immediately obtain the following corollary.
\begin{cor}
There is $R_0>0$ such that for all $R\ge R_0$, 
\[ \liminf_{N\to\infty}\left(-\frac{1}{N}\log\int^\infty_{-\infty}\frac{dE}{\|T_N(E)\|^2}\right)=\liminf_{N\to\infty}\left(-\frac{1}{N}\log\int^R_{-R}\frac{dE}{\|T_N(E)\|^2}\right) \]
holds.
\end{cor}

Let us give the proof of Theorem 3.5 step by step.

Here, let us consider a Schr\"{o}dinger operator $h_\mathbb{Z}$ on a doubly infinite lattice $\mathbb{Z}$. Now, potential is given only on $\mathbb{N}$. For $n=0,-1,\cdots$, we extend it by $v(n)=0$. Then, $h_\mathbb{Z}$ is a bounded self-adjoint operator on $l^2(\mathbb{Z})$ and $\sigma(h_\mathbb{Z})\subset[-C,C]$ ($C=2+\displaystyle\sup_{n\in\mathbb{N}}|v(n)|$). Thus, if $|E|\ge R$, $h_\mathbb{Z}-E$ is invertible. Note that there is a solution $\psi$ of the eigenvalue equation $h_\mathbb{Z}\psi=E\psi$ such that $\psi(n)=\langle e_0,(h-E)^{-1}e_n\rangle$ for $n=0,1,2,\cdots$. Such $\psi$ can be constructed as follows: If $n\in\mathbb{N}$
\begin{eqnarray*}
-\langle e_0,(h_\mathbb{Z}-E)^{-1}e_{n+1}\rangle-\langle e_0,(h_\mathbb{Z}-E)^{-1}e_{n-1}\rangle+v(n)\langle e_0,(h_\mathbb{Z}-E)^{-1}e_n\rangle&=&\langle e_0,(h_\mathbb{Z}-E)^{-1}h_\mathbb{Z}e_n\rangle \\&=&E\langle e_0,(h_\mathbb{Z}-E)^{-1}e_n\rangle 
\end{eqnarray*}
holds. For $n=-1,-2,\cdots$, determine $\psi(n)$ by
\[ \psi(n-1)=-\psi(n+1)+v(n)\psi(n)-E\psi(n) \]
inductively.

Let us consider the asymptotic behavior of $\langle e_0,(h_\mathbb{Z}-E)^{-1}e_n\rangle$. Set $q_n=-|e_n\rangle\langle e_{n-1}|-|e_{n-1}\rangle\langle e_n|$ and $h_n=h_\mathbb{Z}-q_n$. By resolvent formula
\[ (h_\mathbb{Z}-E)^{-1}=(h_n-E)^{-1}-(h_\mathbb{Z}-E)^{-1}q_n(h_n-E)^{-1}, \]
\[ \langle e_0,(h_\mathbb{Z}-E)^{-1}e_n\rangle=\langle e_n,(h_n-E)^{-1}e_n\rangle\langle e_0,(h_\mathbb{Z}-E)^{-1}e_{n-1}\rangle. \]
Use this equation for $\langle e_0,(h_\mathbb{Z}-E)^{-1}e_{n-1}\rangle$ again and repeat this process, then finally we get
\[ \langle e_0,(h_\mathbb{Z}-E)^{-1}e_n\rangle=\langle e_0,(h_\mathbb{Z}-E)^{-1}e_0\rangle\prod^n_{k=1}\langle e_k,(h_k-E)^{-1}e_k\rangle. \]
By spectral decomposition and the condition on $E$, the absolute value of each factor is bounded by $\frac{1}{|E|-C}$. Thus, we have
\[ |\langle e_0,(h_\mathbb{Z}-E)^{-1}e_n\rangle|\le\left(\frac{1}{|E|-C}\right)^{n+1}. \]
Define $\alpha(n)=\psi(n)/\sqrt{\psi(0)^2+\psi(1)^2}$. Let $\beta(n)$ be the solution of the eigenvalue equation with the condition $\beta(0)=-\overline{\alpha(1)},\ \beta(1)=\overline{\alpha(0)}$ ($|\beta(0)|^2+|\beta(1)|^2=1$). By the property of transfer matrix,
\[ \left(\begin{array}{cc}
\alpha(n+1) & \beta(n+1)  \\
\alpha(n) & \beta(n)
\end{array}\right)=T_n(E)\left(\begin{array}{cc}
\alpha(1) & \beta(1)  \\
\alpha(0) & \beta(0)
\end{array}\right) .\]
Since $T_n(E)\in SL(2,\mathbb{C})$ and $\alpha(1)\beta(0)-\alpha(0)\beta(1)=1$, $\alpha(n+1)\beta(n)-\alpha(n)\beta(n+1)=1$ holds. Thus we have
\begin{eqnarray*}
1&\le&|\alpha(n+1)\beta(n)|+|\alpha(n)\beta(n+1)| \\
&\le &\frac{1}{\sqrt{|\psi(0)|^2+|\psi(1)|^2}}\left[\left(\frac{1}{|E|-C}\right)^{n+1}|\beta(n)|+\left(\frac{1}{|E|-C}\right)^n|\beta(n+1)|\right] \\
&\le&\frac{1}{\sqrt{|\psi(0)|^2+|\psi(1)|^2}}\left(\frac{1}{|E|-C}\right)^{n}(|\beta(n)|+|\beta(n+1)|) .
\end{eqnarray*}
By
\[ |\psi(0)|=|\langle e_0,(h-E)^{-1}e_0\rangle|\ge\frac{1}{|E|+C} \]
we get
\[ |\beta(n)|+|\beta(n+1)|\ge\frac{(|E|-C)^n}{|E|+C}. \]
By 
\[ |\beta(n)|^2+|\beta(n+1)|^2\ge\frac{(|\beta(n)|+|\beta(n+1)|)^2}{2}\ge\frac{1}{2}\frac{(|E|-C)^{2n}}{(|E|+C)^2} \]
and $|\beta(0)|^2+|\beta(1)|^2=1$, we get
\[ \|T_n(E)\|\ge\frac{1}{\sqrt{2}}\frac{(|E|-C)^n}{|E|+C}. \]
\[ \int^\infty_R\frac{dE}{\|T_n(E)\|^2}\le2\int^\infty_R\frac{(|E|+C)^2}{(|E|-C)^{2n}}dE=\frac{2}{2n-1}\frac{1}{(R-C)^{2n-1}}+\frac{8C}{2n}\frac{1}{(R-C)^{2n}}+\frac{8C^2}{2n+1}\frac{1}{(R-C)^{2n+1}}. \]
Thus, Theorem 3.5 follows (the case of $\int^{-R}_{-\infty}$ is similarly proven). $\Box$

By this theorem, it turns out that Theorem 1.1 in \cite{bruneau2016conductance} is also true in our setting. We state as a theorem here.
\begin{thm}[\cite{bruneau2016conductance}]
Let $h_\mathbb{N}$ be a discrete Schr\"{o}dinger operator on $l^2(\mathbb{N})$ with a bounded potential $v\colon\mathbb{N}\to\mathbb{R}$. The following statements are equivalent.
\begin{itemize}
    \item $h_\mathbb{N}$ does not have absolutely continuous spectrum ($\sigma_{ac}(h_\mathbb{N})=\emptyset$)
    \item $\displaystyle\lim_{N\to\infty}\int_\mathbb{R}\frac{dE}{\|T_N(E)\|^2}=0$.
\end{itemize}
\end{thm}

\subsubsection{Dynamically defined potentials}
The above results can be applied to arbitrary (bounded) potentials. Next we investigate the detail for a class of potentials called dynamically defined potentials. This class contains various physically important models such as the Anderson model, which is an example of random systems, and the Fibonacci Hamiltonian, which is considered as the one-dimensional model of a quasi-crystal.  There are a huge number of studies for the spectrum of Sch\"{o}dinger operators with dynamically defined potentials \cite{damanik2015spectral,damanik2017schrodinger}. Here, the scaling of the asymptotic behavior is shown to be related with the Lyapunov exponent. 

Let us start with the definition of dynamically defined potentials. We deal with the system on $\mathbb{Z}$, although we are interested in the half of it, $\mathbb{N}$. 

Let $(\Omega,\mathcal{F},P,\phi)$ be an ergodic invertible discrete dynamical system. That is, $(\Omega,\mathcal{F},P)$ is a probability space (in the sequel, we do not write the $\sigma$-field $\mathcal{F}$), $\phi\colon\Omega\to\Omega$ is a measurable bijection preserving probability $P$ such that the probability of invariant set is $0$ or $1$ (ergodicity). Let $f$ be a bounded real measurable function on $\Omega$. Then, for $\omega\in\Omega$ we have a Schr\"{o}dinger operator $h_\omega$ with a potential
\[ v_\omega(n)=f(\phi^n\omega),\ n\in\mathbb{Z}. \]
This $v_\omega(\cdot)$ is called a dynamically defined potential and a family of operators $\{h_\omega\}_{\omega\in\Omega}$ is called an ergodic Schr\"{o}dinger operator.

Let us denote $T_{N,\omega}(E)$ the transfer matrix determined by the potential $v_\omega$. Then $T_{N,\omega}(E)$ satisfies
\[ T_{N+M,\omega}(E)=T_{N,\phi^M\omega}(E)T_{M,\omega}(E) \]
and
\[ \log\|T_{N+M,\omega}(E)\|\le\log\|T_{N,\phi^M\omega}(E)\|+\log\|T_{M,\omega}(E)\| .\]
By subadditive ergodic theorem, for a.e. $\omega$
\[ \lim_{N\to\infty}\frac{1}{N}\log\|T_{N,\omega}\|=L(E) \]
holds, where
\[ L(E)\equiv\inf_{N\ge1}\frac{1}{N}\int_\Omega\log\|T_{N,\omega}(E)\|dP(\omega)=\lim_{N\to\infty}\frac{1}{N}\int_\Omega\log\|T_{N,\omega}(E)\|dP(\omega) .\]
$L(E)$ is called Lyapunov exponent. Since $\|T_{N,\omega}(E)\|\ge1$, $L(E)\ge0$. The Lyapunov exponent $L(E)$ provides a rate of exponential growth of the norm of the transfer matrix $\|T_{N,\omega}(E)\|$ for each $E\in\mathbb{R}$. What we would like to estimate is the integral
\[ \mathcal{I}(N,\omega)\equiv\int^\infty_{-\infty}\frac{dE}{\|T_{N,\omega}(E)\|^2} .\]

\begin{thm}
Assume that the Lyapunov exponent $L(E)$ is continuous. Then,
\begin{eqnarray*}
0&\le&\liminf_{N\to\infty}\left(-\frac{1}{N}\log\mathcal{I}(N,\omega)\right) \\
&\le&\limsup_{N\to\infty}\left(-\frac{1}{N}\log\mathcal{I}(N,\omega)\right) \\
&\le& 2\min_{E\in\mathbb{R}}L(E)
\end{eqnarray*}
holds for a.e. $\omega\in\Omega$.
\end{thm}
\begin{proof}
Only the last inequality is not trivial. Suppose that $\omega\in\Omega$ satisfies
\[ \lim_{N\to\infty}\frac{1}{N}\log\|T_{N,\omega}(E)\|=L(E) \]
for a.e. $E\in\mathbb{R}$. By Fubini theorem, the probability of the set of such $\omega$ is 1. By the discussion of Theorem 3.5 it turns out that $\displaystyle\inf_{E\in\mathbb{R}}L(E)=\min_{E\in\mathbb{R}}L(E)$. Put $\gamma=\displaystyle\min_{E\in\mathbb{R}}L(E)$ and let $E_{min}$ be the energy that achieves the minimum (such $E_{min}$ may not be uniquely determined, but the choice of $E_{min}$ is not important in the following discussion). Since $L(E)$ is continuous, for any $\epsilon>0$ there is $\delta>0$ such that $E\in(E_{min}-\delta,E_{min}+\delta)\equiv R_\delta\Rightarrow L(E)-\gamma < \frac{\epsilon}{2}.$ As $-\log$ is a monotonically decreasing convex function, we have
\begin{eqnarray*}
-\frac{1}{N}\log\mathcal{I}(N,\omega)&\le&-\frac{1}{N}\log\left(\int_{R_\delta}\frac{1}{\|T_{N,\omega}(E)\|^2}dE\right)\\
&=&-\frac{1}{N}\log\left(\frac{1}{2\delta}\int_{R_\delta}\frac{1}{\|T_{N,\omega}(E)\|^2}dE\right)-\frac{1}{N}\log2\delta\\
&\le&\frac{2}{2\delta}\int_{R_\delta}\frac{1}{N}\log\|T_{N,\omega}(E)\|dE-\frac{1}{N}\log2\delta.
\end{eqnarray*}
By dominated convergence theorem,
\[ \limsup_{N\to\infty}\left(-\frac{1}{N}\log\mathcal{I}(N,\omega)\right)\le\frac{2}{2\delta}\int_{R_\delta}L(E)dE\le2\gamma+\epsilon. \]
Since $\epsilon>0$ is arbitrary, we get
\[ \limsup_{N\to\infty}\left(-\frac{1}{N}\log\mathcal{I}(N,\omega)\right)\le2\min_{E\in\mathbb{R}}L(E). \]
\end{proof}
By this theorem, if the Lyapunov exponent $L(E)$ is continuous and $\displaystyle\min_{E\in\mathbb{R}}L(E)=0$, the current does not decay exponentially. Examples are given in the last of this section. Although this theorem tells when the decay of the current is slow, it does not tell when the current decays exponentially. We do not know whether the equality holds or not in Theorem 3.8. If the following large deviation type estimate and $\displaystyle\inf_{E\in\mathbb{R}}L(E)>0$ are given, we can conclude the exponential decay of the current.

\begin{dfn}[Large Deviation type estimate]
We say that the property LD (Large Deviation type estimate) holds, if the following condition is satisfied: 
For any $\epsilon>0$ and any finite closed interval $[a,b]$, there are constants $C,\eta>0$ such that 
\[ P\left(\left\{\omega\in\Omega\mid\left|\frac{1}{N}\log\|T_{N,\omega}(E)\|-L(E)\right|\ge\epsilon \right\}\right)\le Ce^{-\eta N},\ ^\forall N\in\mathbb{N},^\forall E\in[a,b]. \]
\end{dfn}

\begin{thm}
Suppose that the property LD holds and $\displaystyle\inf_{E\in\mathbb{R}}L(E)>0$, then
\[ \liminf_{N\to\infty}\left(-\frac{1}{N}\log\mathcal{I}(N,\omega)\right)>0,\ a.e.\ \omega\in\Omega.\]
\end{thm}
Although the proof is obvious from the discussion in the proof of Lemma 3.2 in \cite{jitomirskaya2019large}, we repeat it here.
\begin{proof}
Set $\gamma=\displaystyle\inf_{E\in\mathbb{R}}L(E)>0$ and fix $\epsilon,R$ that satisfy $0<\epsilon<\gamma$ and $R>3+\|f\|$ ($\|f\|$ is the norm in $L^\infty(\Omega,P)$). By the property LD, there are $\eta,C>0$ such that
\[ P\left(\left\{\omega\in\Omega\mid\left|\frac{1}{N}\log\|T_{N,\omega}(E)\|-L(E)\right|\ge\epsilon \right\}\right)\le Ce^{-\eta N},\ ^\forall N\in\mathbb{N},^\forall E\in[-R,R]. \]
Let us denote $m$ Lebesgue measure on $\mathbb{R}$. Denote
\[ \Omega_\epsilon^N=\left\{(E,\omega)\in [-R,R]\times\Omega\mid\left|\frac{1}{N}\log\|T_{N,\omega}(E)\|-L(E)\right|\ge\epsilon \right\} ,\]
\[ \Omega_\epsilon^N(\omega)=\{E\in [-R,R]\mid(E,\omega)\in\Omega_\epsilon^N\}, \]
then we have
\[ m\times P(\Omega_\epsilon^N)\le2RCe^{-\eta N}. \]
Fix $\delta$ such that $0<\delta<\eta$ and set
\[ X_\delta^N=\{\omega\in\Omega\mid m(\Omega_\epsilon^N(\omega))\le e^{-\delta N} \}, \]
then we get
\begin{eqnarray*}
P(X_\delta^{N,C})&\le& e^{\delta N}\int_{X_\delta^{N,C}}m(\Omega_\epsilon^N(\omega))P(d\omega) \\
&\le&e^{\delta N}m\times P(\Omega_\epsilon^N) \\
&\le&2RCe^{-(\eta-\delta)N}.
\end{eqnarray*}

\[ \sum_{N=1}^\infty P(X_\delta^{N,C})<\infty \]
holds and by Borel-Cantelli lemma,
\[ P\left(\liminf_{N\to\infty}X_\delta^N\right)=1 .\]
This means that for a.e. $\omega$ there is $N(\omega)\in\mathbb{N}$ such that if $N\ge N(\omega)$ then $m(\Omega_\epsilon^N(\omega))\le e^{-\delta N}$ holds. Obviously such $\omega$ satisfies 
\begin{eqnarray*}
\int^R_{-R}\frac{dE}{\|T_{N,\omega}(E)\|^2}&\le&\int_{\Omega_\epsilon^N(\omega)}\frac{dE}{\|T_{N,\omega}(E)\|^2}+\int_{\Omega_\epsilon^N(\omega)^C} \frac{dE}{\|T_{N,\omega}(E)\|^2} \\
&\le& e^{-\delta N}+\int^R_{-R}\frac{dE}{e^{2(L(E)-\epsilon)N}} \\
&\le& e^{-\delta N}+2Re^{-2(\gamma-\epsilon)N}
\end{eqnarray*}
for $N\ge N(\omega)$. By this estimate and Theorem 3.5, we obtain
\[ \liminf_{N\to\infty}\left(-\frac{1}{N}\log\left(\int_\mathbb{R}\frac{dE}{\|T_{N,\omega}(E)\|^2}\right)\right)\ge\min\{\delta,2(\gamma-\epsilon),2\log(R-C)\}>0 .\]
\end{proof}

\subsubsection{Examples}
The continuity and the Large deviation type estimate of the Lyapunov exponent are already well investigated in the context of ergodic Schr\"{o}dinger operators \cite{bourgain2002continuity}. Here we show some physically important examples. See \cite{duarte2016lyapunov} for well-organized results for the continuity and the large deviation type estimate of the Lyapunov exponent. Here we would like to show some examples.

\underline{The Anderson model}\\
Let $K\subset\mathbb{R}$ be a compact subset, $\rho$ be a probability measure on $K$ such that $\#\mathrm{supp}\rho\ge2$ ($\#$ is the number of elements of the set). Define $\Omega=K^\mathbb{Z}$ and $P=\rho^\mathbb{Z}$. Let $\phi$ be a shift on $\Omega$, that is, $(\phi\omega)_n=\omega_{n+1}$. $f(\omega)=\omega_0$. This is a model such that the value of the potential at each site is the i.i.d. random variable. As is well known, this model exhibits Anderson localization. The following theorem is a statement called spectral localization \cite{carmona1987}.
\begin{thm}
For a.e. $\omega\in\Omega$, the following statements hold:
\begin{itemize}
    \item $h_\omega$ has pure point spectrum.
    \item Every eigenvector decays exponentially.
\end{itemize}
\end{thm}

By Theorem 3.7, the current converges to $0$ as $N\to\infty$ for a.e. $\omega$ (we can apply Theorem 3.7 for the system on $\mathbb{Z}$, since absolutely continuous spectrum is stable under trace class perturbations). 
Moreover, since the Lyapunov exponent satisfies the large deviation type estimate and $\displaystyle\inf_{E\in\mathbb{R}}L(E)>0$ \cite{bucaj2019localization}, the current decays exponentially by Theorem 3.9.

\underline{The Fibonacci Hamiltonian}\\
This model was introduced in \cite{kohmoto1983localization,ostlund1983one} and has been studied as a model of a one-dimensional quasi-crystal. See \cite{damanik2016fibonacci} for detail. The Fibonacci Hamiltonian is defined as follows: 
$\Omega=\mathbb{T}$, $P$ : Lebesgue measure. $\phi\omega=\omega+\alpha$, where $\alpha=\frac{\sqrt{5}-1}{2}$. $f(\omega)=-\lambda\chi_{[1-\alpha,1)}(\omega)$.

The spectrum is independent of $\omega\in\mathbb{T}$ (we denote it by $\Sigma_\lambda$) and singular continuous. It is known that the Lyapunov exponent $L(E)$ is continuous and is $0$ on $\Sigma_\lambda$.  Thus by Theorem 3.7, 3.8, although the current converges to $0$ as $N\to\infty$, it does not decay exponentially. The more can be said for this model. In the case where $\omega=0$, it is shown that the norm of the transfer matrix is bounded by the power of the sample size $N$ on the spectrum \cite{iochum1991power} : There is an $E$-independent constant $\theta>0$ such that if $E\in\Sigma_\lambda$ then $\|T_N(E)\|\le N^\theta$. Note that this fact does not imply the power law decay of the current immediately,
because the Lebesgue measure of the spectrum $\Sigma_\lambda$ is $0$. However, by combining the results in \cite{damanik2008fractal,sutHo1987spectrum}, we can conclude the power law decay of the current.

\begin{thm}
Let $\mathrm{dim}_H\Sigma_\lambda$ be the Hausdorff dimension of $\Sigma_\lambda$ ( $\mathrm{dim}_H\Sigma_\lambda\in(0,1)$ by \cite{damanik2016fibonacci}). For any $\xi\in(0,\mathrm{dim}_H\Sigma_\lambda)$, there is a constant $C_{\xi}>0$ such that
\[ \mathcal{I}(N)\ge \frac{C_\xi}{N^{(\frac{1}{\xi}-1)+2\theta}}. \]
\end{thm}

\underline{Almost Mathieu operator}\\
This model is the representative example of quasi-periodic potential. 
$\Omega=\mathbb{T}$, $P$ : Lebesugue measure. $\phi\omega=\omega+\alpha$ for fixed $\alpha\in\mathbb{T}$. $f(\omega)=-2\lambda\cos(2\pi\omega)$. This model has two parameters $\alpha\in\mathbb{T},\lambda>0$, and the properties vary according to them. Since if $\alpha$ is rational, the porential is periodic, we assume that $\alpha$ is irrational. If $\lambda<1$, then for every $\omega\in\mathbb{T}$ the spectrum of $h_\omega$ is purely absolutely continuous. If $\lambda\ge1$, then for every $\omega\in\mathbb{T}$, absolutely continuous spectrum is empty, $\sigma_{ac}(h_\omega)=\emptyset$. So our interest is in the case where $\lambda\ge1$. The Lyapunov exponent $L(E)$ is continuous and its minimum is $\max\{\log\lambda,0\}$, which is the value on the spectrum \cite{bourgain2002continuity}. Thus, the current does not show the exponential decay for $\lambda=1$. If $\lambda>1$, it is shown that the property LD holds for appropriate $\alpha$, and the current decays exponentially
\cite{goldstein2001holder}.

\subsection{$\beta>0$ : with noise}
In this subsection we consider the current under dephasing noise. We obtain an explicit form of the current, which scales as $1/N$ for large $N$, in the case where the potential is absent (3.2.1). 3.2.2 deals with the general potential case. Unfortunately, the scaling of the current for general potentials is not obtained yet. But we can say a little about the current for strong noise regime.

\subsubsection{$v=0$}
Let us start with the case where $v=0$. In this case we can obtain an explicit form of the current $\mathcal{J}_\beta(N)$, using the equation
\[ \mathcal{J}_\beta(N)=-4(\alpha_{in}^l\alpha_{out}^r-\alpha_{out}^l\alpha_{in}^r)\langle e_1,l^{-1}(p_N)e_1\rangle .\]

Set $X=l^{-1}(p_N)$. $X$ is a self-adjoint operator on $\mathbb{C}^N$. Let us denote $X_{ij}=\langle e_i,Xe_j\rangle$. Since $X$ is self-adjoint, $X_{ji}=\overline{X_{ij}}$. Denote $\alpha_{in}^l+\alpha_{out}^l=\zeta_l>0,\ \alpha_{in}^r+\alpha_{out}^r=\zeta_r>0$. By $l(X)=p_N$, we have
\[ 0=\langle e_1,p_Ne_1\rangle=\langle e_1,l(X)e_1\rangle=-2\zeta_l X_{11}+iX_{21}-iX_{12} \]
\[ \to\ \mathrm{Im}X_{12}=\zeta_l X_{11}. \]
And for $n=2,\cdots,N-1$,
\[ 0=\langle e_n,p_Ne_n\rangle=iX_{n-1n}-iX_{nn-1}+iX_{n+1n}-iX_{nn+1} \]
\[ \to\ \mathrm{Im}X_{n-1n}=\mathrm{Im}X_{nn+1}=\zeta_l X_{11} .\]
By 
\[ \int^\infty_0T_t^*(2\zeta_l p_1+2\zeta_r p_N)dt=I ,\]
where $T_t^*$ is the dual action of $T_t$ ($\mathrm{Tr}aT_t(b)=\mathrm{Tr}T_t^*(a)b$), we get
\[ 2\zeta_l X_{11}+2\zeta_r X_{NN}=-1 \]
\[ \to X_{NN}=\left(-\frac{1}{2\zeta_r}-\frac{\zeta_l}{\zeta_r}X_{11}\right). \]
We have
\[ 0=\langle e_1,p_Ne_2\rangle=-\zeta_l X_{12}-\beta X_{12}+iX_{22}-iX_{11}-iX_{13}, \]
\[ 0=\langle e_{N-1},p_Ne_N\rangle=-\zeta_r X_{N-1N}-\beta X_{N-1N}+iX_{NN}+iX_{N-2N}-iX_{N-1N-1}, \]
and for $n=2,\cdots,N-2$
\[ 0=\langle e_n,p_Ne_{n+1}\rangle=-\beta X_{nn+1}+iX_{n+1n+1}+iX_{n-1n+1}-iX_{nn}-iX_{nn+2}. \]
Adding the imaginary part of the above three equations, we finally obtain
\begin{eqnarray*}
 0&=&X_{NN}-X_{11}-\zeta_l\cdot\zeta_l X_{11}-\zeta_r\cdot\zeta_l X_{11}-\beta(N-1)\cdot\zeta_l X_{11} \\
&=&\left(-\frac{1}{2\zeta_r}-\frac{\zeta_l}{\zeta_r}X_{11}\right)-X_{11}-2\zeta_l^2X_{11}-2\zeta_l\zeta_r X_{11}-\beta\zeta_l(N-1)X_{11}. 
\end{eqnarray*}
\[ \to X_{11}=-\frac{1}{2}\frac{1}{\zeta_l+\zeta_r+\zeta_l\zeta_r(\zeta_l+\zeta_r+\beta(N-1))} .\]
Thus the current $\mathcal{J}_\beta(N)$ is expressed as follows:
\begin{thm}
When $v=0$, then
\[ \mathcal{J}_\beta(N)=\frac{2(\alpha_{in}^l\alpha_{out}^r-\alpha_{out}^l\alpha_{in}^r)}{\alpha_{in}^l+\alpha_{out}^l+\alpha_{in}^r+\alpha_{out}^r+(\alpha_{in}^l+\alpha_{out}^l)(\alpha_{in}^r+\alpha_{out}^r)(\alpha_{in}^l+\alpha_{out}^l+\alpha_{in}^r+\alpha_{out}^r+\beta(N-1))} .\]
\end{thm}
The current $\mathcal{J}_\beta(N)$ decays as $1/N$ for large $N$ and its coefficient is
\[ \frac{2(\alpha_{in}^l\alpha_{out}^r-\alpha_{out}^l\alpha_{in}^r)}{\beta(\alpha_{in}^l+\alpha_{out}^l)(\alpha_{in}^r+\alpha_{out}^r)} .\]
For $\alpha_{in}^l=\Gamma\frac{1-\mu}{2},\ \alpha_{out}^l=\Gamma\frac{1+\mu}{2},\ \alpha_{in}^r=\Gamma\frac{1+\mu}{2},\ \alpha_{out}^r=\Gamma\frac{1-\mu}{2} ,\ \beta=2\gamma$, we have
\[ \mathcal{J}_{2\gamma}(N)=-\frac{\mu}{\Gamma+1/\Gamma+\gamma(N-1)}. \]
This corresponds to the result of \cite{vznidarivc2010exact} (note that the Hamiltonian in \cite{vznidarivc2010exact} corresponds to $2H$ in our setting).

\subsubsection{$v$ : general potentials}
In the case of general potentials, the scaling of $\mathcal{J}_\beta(N)$ is not obtained. But for large $\beta$, we can know a little about the current. First, we consider the strong noise limit $\beta\to\infty$. And then, large but finite noise $\beta=\epsilon N$ is discussed and it is shown that the current may be increased by adding large noise in the case of random potentials. 

The same calculation as the case where $v=0$ shows that
\[ [\zeta_l+\zeta_r+\zeta_l\zeta_r(\zeta_l+\zeta_r+\beta(N-1))]X_{11}=-\frac{1}{2}+\zeta_r\sum_{n=1}^{N-1}(v(n+1)-v(n))\mathrm{Re}X_{nn+1}.  \]
Since $X$ is bounded:
\[ 0\le-X=-l^{-1}(p_N)=\int^\infty_0e^{tl}(p_N)dt\le\frac{1}{\zeta_r}I ,\]
we have
\[ |X_{nn+1}|\le\frac{1}{\zeta_r\beta}(3+\max\{\zeta_l,\zeta_r,1\})\to0\ (\beta\to\infty). \]
Thus we obtain
\[ \lim_{\beta\to\infty}\beta\mathcal{J}_\beta(N)=\frac{2(\alpha_{in}^l\alpha_{out}^r-\alpha_{out}^l\alpha_{in}^r)}{(\alpha_{in}^l+\alpha_{out}^l)(\alpha_{in}^r+\alpha_{out}^r)(N-1)}. \]
This means that when one expands $\mathcal{J}_\beta(N)$ in terms of $1/\beta$ for large $\beta$, the dominant term is 
\[ \frac{2(\alpha_{in}^l\alpha_{out}^r-\alpha_{out}^l\alpha_{in}^r)}{(\alpha_{in}^l+\alpha_{out}^l)(\alpha_{in}^r+\alpha_{out}^r)(N-1)}\frac{1}{\beta}, \]
which is independent of potentials and scales $1/N$ for large $N$. But there is a gap between this fact and the claim that $\mathcal{J}_\beta(N)$ scales as $1/N$.

Next, we consider large $\beta$ not taking limit $\beta\to\infty$. Denote $C=3+\max\{\zeta_l,\zeta_r,1\}$. Fix $\epsilon>4\|v\|C$ and put $\beta=\epsilon N$, then we have
\begin{eqnarray*}
[\zeta_l+\zeta_r+\zeta_l\zeta_r(\zeta_l+\zeta_r+\beta(N-1))]X_{11}&=&-\frac{1}{2}+\zeta_r\sum_{n=1}^{N-1}(v(n+1)-v(n))\mathrm{Re}X_{nn+1}\\
&\le&-\left(\frac{1}{2}-\frac{2\|v\|C}{\epsilon}\right) .
\end{eqnarray*}
Therefore, the current $\mathcal{J}_\beta(N)$ is bounded below as
\[ \mathcal{J}_{\epsilon N}(N)\ge\frac{4(\alpha_{in}^l\alpha_{out}^r-\alpha_{out}^l\alpha_{in}^r)}{\zeta_l+\zeta_r+\zeta_l\zeta_r(\zeta_l+\zeta_r+\epsilon N(N-1))}\left(\frac{1}{2}-\frac{2\|v\|C}{\epsilon}\right)>0 .\]
Let us consider the Anderson model as an example. Recall that if $\beta=0$, the current shows the exponential decay for a.e. $\omega$. It turns out that by the above inequality, for such $\omega$, 
\[ \mathcal{J}_{\epsilon N}(N,\omega)\ge\mathcal{J}_0(N,\omega) \]
holds for sufficiently large $N$. Thus, strong noise increases the current in this example. It is remarkable that although the noise is symmetric and does not have the effect to flow the particles to a specific direction, it could increase the current. Note that the noise does not always increase the current (consider the case where $v=0$). 

\section{$d$-dimensional systems}
In the previous sections we focused on one-dimensional systems. In this section we consider an extension to general $d$-dimensional systems. As in the one-dimensional case, we assume that particles go in and out in a specific direction. Although the case where $d=2,3$ is physically important, we discuss general $d$-dimensional systems here. Since the analysis is almost the same as one-dimensional systems, we do not discuss the detail here.

For $N_1,N_2,\cdots,N_d\in\mathbb{N}$, let us consider a finite $d$-dimensional lattice
\[ \mathfrak{L}=\{1,2,\cdots,N_1\}\times\{1,2,\cdots,N_2\}\times\{1,2,\cdots,N_d\} .\]
An element of this lattice is written as 
\[ \nu=(\nu_1,\nu_2,\cdots,\nu_d)\in\mathfrak{L}. \]
We assume that particles go in and out in the direction'1'. For $i=1,2,\cdots,N_1$, define
\[ M_i=\{\nu\in\mathfrak{L}\mid\nu_1=i\} .\]
This is a plane vertical to the direction'1'. Suppose that particles go in and out at the surfaces $M_1,M_{N_1}$. For $\nu\in\mathfrak{L}\setminus M_{N_1}$, define
\[ \nu_+=(\nu_1+1,\nu_2,\cdots,\nu_d)\in\mathfrak{L} .\]
And let $NN(\nu)$ be the set of nearest-neighbors of $\nu$ in $\mathfrak{L}$.

one-particle Hilbert space that describes Fermi particles moving on this lattice is $\mathbb{C}^{|\mathfrak{L}|}$, where $|\mathfrak{L}|=\displaystyle\prod_{n=1}^dN_n$. We denote its standard basis by $\{e_\nu\}_{\nu\in\mathfrak{L}}$. one-particle Hamiltonian $h$ is given as
\[ (h\psi)(\nu)=-\sum_{\mu\in NN(\nu)}\psi(\mu)+v(\nu)\psi(v),\ \psi\in\mathbb{C}^{|\mathfrak{L}|} .\]
Let $H$ be the total Hamiltonian constructed by this one-particle Hamiltonian $h$. Let us consider the following generator $L$ in many body system:
\begin{eqnarray*}
L(A)&=&i[H,A]\\
&&+\alpha_{in}^l\sum_{\nu\in M_1}(2a_\nu\theta(A)a^*_\nu-\{a_\nu a^*_\nu,A\}) +\alpha_{out}^l\sum_{\nu\in M_1}(2a^*_\nu\theta(A)a_\nu-\{a^*_\nu a_\nu,A\}) \\
&&+\alpha_{in}^r\sum_{\nu\in M_{N_1}}(2a_\nu\theta(A)a^*_\nu-\{a_\nu a^*_\nu,A\}) +\alpha_{out}^r\sum_{\nu\in M_{N_1}}(2a^*_\nu\theta(A)a_\nu-\{a^*_\nu a_\nu,A\}) \\
&&+\beta\sum_{\nu\in\mathfrak{L}}\left(a^*_\nu a_\nu Aa^*_\nu a_\nu-\frac{1}{2}\{a^*_\nu a_\nu,A\}\right).
\end{eqnarray*}
Here, we denote $a^\#(e_\nu)=a^\#_\nu$ as usual. $\alpha_{in}^l,\alpha_{out}^l,\alpha_{in}^r,\alpha_{out}^r,\beta$ are real numbers that are greater than or equal to $0$, and we assume that $\alpha_{in}^l+\alpha_{out}^l>0,\ \alpha_{in}^r+\alpha_{out}^r>0$. By the same calculation as one-dimensional case, it turns out that the dynamics of the two point function is described in terms of that of one-particle system. For $e_\nu\in\mathbb{C}^{|\mathfrak{L}|}$, denote a 1-rank projection by $p_\nu=|e_\nu\rangle\langle e_\nu|$. If
\[ \omega(a^*(f)a(g))=\langle g,Rf\rangle ,\]
then $R(t)$ defined by the relation
\[ \omega\circ e^{tL}(a^*(f)a(g))=\langle g,R(t)f\rangle, \]
is expressed as 
\[ R(t)=e^{tl}(R)+\int^t_0e^{sl}\left(2\alpha_{in}^l\sum_{\nu\in M_1}p_\nu+2\alpha_{in}^r\sum_{\nu\in M_{N_1}}p_\nu\right)ds ,\]
where $l$ is a linear map on $M_{|\mathfrak{L}|}(\mathbb{C})$ defined as
\[ l(a)=-i[h,a]-(\alpha_{in}^l+\alpha_{out}^l)\left\{\sum_{\nu\in M_1}p_\nu,a\right\}-(\alpha_{in}^r+\alpha_{out}^r)\left\{\sum_{\nu\in M_{N_1}}p_\nu,a\right\}+\beta\left(\sum_{\nu\in\mathfrak{L}}p_\nu ap_\nu-a\right). \]
It generates a semigroup of CP maps $e^{tl}$. By the same discussion as the one-dimensional system, we obtain $\displaystyle\lim_{t\to\infty}e^{tl}=0$. Thus $R(t)$ converges to
\[ \int^\infty_0e^{tl}(2\alpha_{in}^l P_1+2\alpha_{in}^r P_{N_1})dt \]
as $t\to\infty$, where $P_1=\displaystyle\sum_{\nu\in M_1}p_\nu$ and $P_{N_1}=\displaystyle\sum_{\nu\in M_{N_1}}p_\nu$. In the long time limit, the number of particles which move from $M_n$ to $M_{n+1}$ per time (current) becomes
\[ \sum_{\nu\in M_n}\mathrm{Im}\int^\infty_0\left\langle e_{\nu_+},e^{tl}(2\alpha_{in}^l P_1+2\alpha_{in}^r P_{N_1})e_{\nu}\right\rangle dt .\]
It is independent of $n$ (we denote it by $\mathcal{J}(N_1,\cdots,N_d)$). The same calculation as one-dimensional system shows that
\[ \mathcal{J}(N_1,\cdots,N_d)=
4(\alpha_{in}^l\alpha_{out}^r-\alpha_{out}^l\alpha_{in}^r)\int^\infty_0\mathrm{Tr}P_1e^{tl}(P_{N_1})dt=-4(\alpha_{in}^l\alpha_{out}^r-\alpha_{out}^l\alpha_{in}^r)\mathrm{Tr}P_1l^{-1}(P_{N_1}). \]

In the case where $v=0$, we obtain the explicit form of the current:
\begin{thm}
\begin{eqnarray*}
&&\mathcal{J}(N_1,\cdots,N_d)=\\
&&\frac{2(\alpha_{in}^l\alpha_{out}^r-\alpha_{out}^l\alpha_{in}^r)\prod_{n=2}^dN_n}{(\beta(N_1-1)+\alpha_{in}^l+\alpha_{out}^l+\alpha_{in}^r+\alpha_{out}^r)(\alpha_{in}^l+\alpha_{out}^l)(\alpha_{in}^r+\alpha_{out}^r)+\alpha_{in}^l+\alpha_{out}^l+\alpha_{in}^r+\alpha_{out}^r} .
\end{eqnarray*}
\end{thm}
Especially in the case where $d=3$, the current decreases in inverse proportion to the length of the sample $N_1$ and increases in proportion to the cross section $N_2\times N_3$.

\section{Discussion and conclusions}
In this paper, we investigated the current for a conduction model of Fermi particles on a finite lattice. When the dephasing noise is absent ($\beta=0$), this model is a special case of those in \cite{prosen2008third,davies1977}. First, we obtained the dynamics of two point function and proved that it converges to a constant independent of initial state. Next, we investigated the current, which is an important quantity in nonequilibrium systems and described by two point function and obtained a simple current formula (Theorem 2.2). Based on this formula, we considered the asymptotic behavior of the current. The results are as follows: 
\begin{description}
\item[noiseless ($\beta=0$)]One can evaluate the current using transfer matrix. For dynamically defined potentials, the asymptotic behavior is related to the property of the Lyapunov exponent. For example, the Anderson model shows the exponential decay of current.
\item[with noise ($\beta>0$)]For the case where $v=0$, the current is explicitly obtained and decays as $1/N$. The same analysis can be applied to higher dimensional systems. In three-dimensional case, the current increases in proportion to cross section and decreases in inverse proportion to the length of the sample for large sample size.
\end{description}
Apart from the case where $v=0$, we gave only inequalities for the asymptotic property in this paper. To obtain the exact scaling of the current for various models is our future work.

Finally we would like to discuss some related studies. As previously mentioned, the noiseless case is also studied in more general settings in \cite{prosen2008third,davies1977}. But we believe that it is our original work to obtain the current formula (Theorem 2.2) and investigate the asymptotic property based on it. In \cite{prosen2008third}, Prosen discussed the conduction model as an example and said that the current would decay exponentially for random potentials. But he did not give an exact proof for it. The model that noise exists and potential $v=0$ is studied in \cite{vznidarivc2010exact}, and the same current formula as ours (subsection 3.2.1) is obtained for special values of $\alpha_{in}^l,\alpha_{out}^l,\alpha_{in}^r,\alpha_{out}^r$. However, the approach is different from ours. We solved the time evolution of the current and showed that the current converges to a stable value independent of initial states. On the other hand, in \cite{vznidarivc2010exact} {\v{Z}}nidari{\v{c}} tried to obtain a nonequilibrium stationary state directly as a state $\rho$ which satisfies $L(\rho)=0$. Since he obtained a stationary state based on an ansatz, it is not obvious if this state is the unique stationary state and the system converges to it (and if 'the stationary state' he obtained satisfies the condition of state, $\rho\ge0$). And general potential case and higher dimensional case are not discussed in \cite{vznidarivc2010exact}.

The model discussed in this paper is described by a finite dimensional open system. As mentioned in 1 Introduction, there is a different approach that considers the Hamiltonian dynamics of the total system including infinitely extended reservoirs \cite{bruneau2016conductance}. In their model, the current in nonequilibrium stationary state is evaluated by
\[ \int^{\mu_L}_{\mu_R}\frac{dE}{\|T_N(E)\|^2} \]
\cite{bruneau2016conductance,bruneau2016absolutely}, where $\mu_L,\mu_R\ (\mu_L>\mu_R)$ are chemical potentials of the reservoirs. The difference between our model and this model is only the region of integral, one is $\mathbb{R}$ and the other is $[\mu_R,\mu_L]$. But by Theorem 3.4, if $[\mu_R,\mu_L]$ is sufficiently large, this difference does not matter and both model give the same prediction for the asymptotic behavior.

\bibliography{transport}
\bibliographystyle{unsrt}

\end{document}